\documentclass[conference]{IEEEtran}
\IEEEoverridecommandlockouts
\usepackage{array}
\usepackage{amsmath}
\usepackage{float}
\usepackage{multirow}
\usepackage{array}
\usepackage{cancel}
\usepackage{tabu}
\usepackage{amsthm}
\usepackage{epsfig}
\usepackage{epstopdf}
\usepackage{epsf}
\usepackage{color}
\usepackage[english]{babel}

\theoremstyle{definition}
\newtheorem{theorem}{Theorem}
\newtheorem{lemma}[theorem]{Lemma}
\newtheorem{definition}{Definition}

\usepackage{amssymb}
\usepackage{cite}
\usepackage{amsmath,amssymb,amsfonts}
\usepackage{textcomp}
\usepackage{xcolor}
\usepackage{algorithm}
\usepackage[noend]{algpseudocode}
\usepackage[T1]{fontenc}
\usepackage{enumitem}
\usepackage{lipsum}
\usepackage{mathtools}
\usepackage{cuted}

\usepackage{optidef}
\usepackage{graphicx}
\usepackage[caption=false,font=footnotesize]{subfig}
\usepackage{hyperref}

\begin{document}
    
	\title{Rate-Distortion-Classification Representation Theory for Bernoulli Sources}
    

    \author{\IEEEauthorblockN{$\textnormal{Nam Nguyen}$, $\textnormal{Thinh Nguyen}$, and $\textnormal{Bella Bose}$ \thanks{This work was supported by the National Science Foundation under Grant No. CCF:SHF:2417898.}}
    \IEEEauthorblockA{$\textnormal{School of Electrical and Computer Engineering, Oregon State University, Corvallis, OR, 97331}$\\}
    Emails: nguynam4@oregonstate.edu, thinhq@eecs.oregonstate.edu, bella.bose@oregonstate.edu}

\maketitle
\begin{abstract}
We study task-oriented lossy compression through the lens of rate-distortion-classification (RDC) representations. The source is Bernoulli, the distortion measure is Hamming, and the binary classification variable is coupled to the source via a binary symmetric model. Building on the one-shot common-randomness formulation, we first derive closed-form characterizations of the one-shot RDC and the dual distortion-rate-classification (DRC) tradeoffs. We then use a representation-based viewpoint and characterize the achievable distortion-classification (DC) region induced by a fixed representation by deriving its lower boundary via a linear program. Finally, we study universal encoders that must support a family of DC operating points and derive computable lower and upper bounds on the minimum asymptotic rate required for universality, thereby yielding bounds on the corresponding rate penalty. Numerical examples are provided to illustrate the achievable regions and the resulting universal RDC/DRC curves.
\end{abstract}

\renewcommand\IEEEkeywordsname{Keywords}

\section{Introduction}
Rate-distortion theory provides a fundamental framework for understanding the limits of lossy compression by characterizing the minimum rate required to represent a source subject to a fidelity constraint~\cite{berger1971ratedistortion,Cover2006}. While classical rate-distortion theory focuses on reconstruction accuracy, it is now well understood that minimizing distortion alone does not necessarily preserve information relevant for downstream tasks or human perception.

This observation has motivated extensive research on task-aware and perception-constrained lossy compression. In particular, Blau and Michaeli~\cite{blau2019rethinking} introduced the rate-distortion-perception (RDP) framework, which incorporates perceptual quality as a third dimension alongside rate and distortion. By modeling perceptual quality through distribution-level constraints between the source and reconstruction, the RDP framework reveals a fundamental tradeoff between distortion fidelity and perceptual realism. These ideas have inspired both theoretical developments and practical deep learning-based compression systems~\cite{blau2018perception,liu2022lossy,larsen2016autoencoding,agustsson2019generative,mentzer2020high,yang2023lossy}.

More recently, attention has shifted toward compression schemes that explicitly account for downstream inference and decision-making. In such settings, the reconstruction is required not only to approximate the source, but also to preserve information relevant to a latent task variable~\cite{tishby2000information, wu2016information, stavrou2023role}. The rate-distortion-classification (RDC) framework formalizes this perspective by introducing information-theoretic constraints that quantify the residual uncertainty of the task variable given the reconstruction~\cite{CDP,Zhang2023,Wang2024}. Compared to perception-based constraints, classification constraints directly capture task performance and are particularly well suited for decision-oriented systems~\cite{mentzer2023lossy,Zhongyue2023,Lei2023,Luo2021,nguyen2026cross}.

While existing RDC results characterize optimal tradeoffs for individual operating points, practical systems often require a single encoder to support multiple distortion and classification requirements. Redesigning or retraining an encoder for each operating point is undesirable in resource-constrained or multi-task environments. This motivates the notion of \emph{universal encoder representations}, where a fixed encoder produces a shared representation that can be paired with different decoders to meet a family of objectives. The fundamental cost of enforcing such universality is captured by the associated \emph{rate penalty}, defined as the excess rate required by a universal encoder relative to the best task-specific design.

A universal encoder representation-based framework for perception-constrained lossy compression was developed by Zhang \emph{et al.}~\cite{UniversalRDPs}, where universal RDP representations and one-size-fits-all encoders were formalized, along with the associated rate penalty. Building on this framework, Qian \emph{et al.}~\cite{Qian2022} specialized the theory to Bernoulli sources under Hamming distortion, enabling closed-form characterizations of the one-shot RDP function and a complete description of achievable distortion-perception regions. Extending the notion of universality from perception to task-oriented settings, Nguyen \emph{et al.}~\cite{Nam2025} developed the theoretical framework for universal RDC representations, establishing fundamental limits on reusing a single encoder across multiple distortion-classification (DC) objectives.

In this paper, we study RDC representations for Bernoulli sources under Hamming distortion, where a binary classification variable is coupled to the source through a symmetric noise model. Building on the binary-source perception framework in~\cite{Qian2022}, we extend the analysis to classification constraints. We first derive closed-form characterizations of the one-shot RDC and distortion-rate-classification (DRC) tradeoffs under common randomness and derive the achievable DC region induced by a fixed representation. We then consider the universal setting, in which a single encoder must support a family of DC operating points, and derive computable lower and upper bounds on the resulting minimum rate penalty.

\section{RDC Tradeoff with Common Randomness}

\subsection{Problem Formulation}
We study a task-oriented lossy compression problem in which an observable source
$X \sim p_X$ is statistically coupled with a collection of latent target variables
$S_1,\dots,S_K \sim p_{S_1,\dots,S_K}$.
The joint distribution $p_{X,S_1,\dots,S_K}$ captures the dependence between the observable data and the underlying task variables.
Although the variables $S_1,\dots,S_K$ are not directly observed, they are typically recoverable, possibly with uncertainty, from the source observation $X$.
For example, when $X$ represents an image, the variables $S_k$ may correspond to semantic attributes such as object classes or scene categories.

\begin{figure}[h]
\centering
\includegraphics[width=0.48\textwidth]{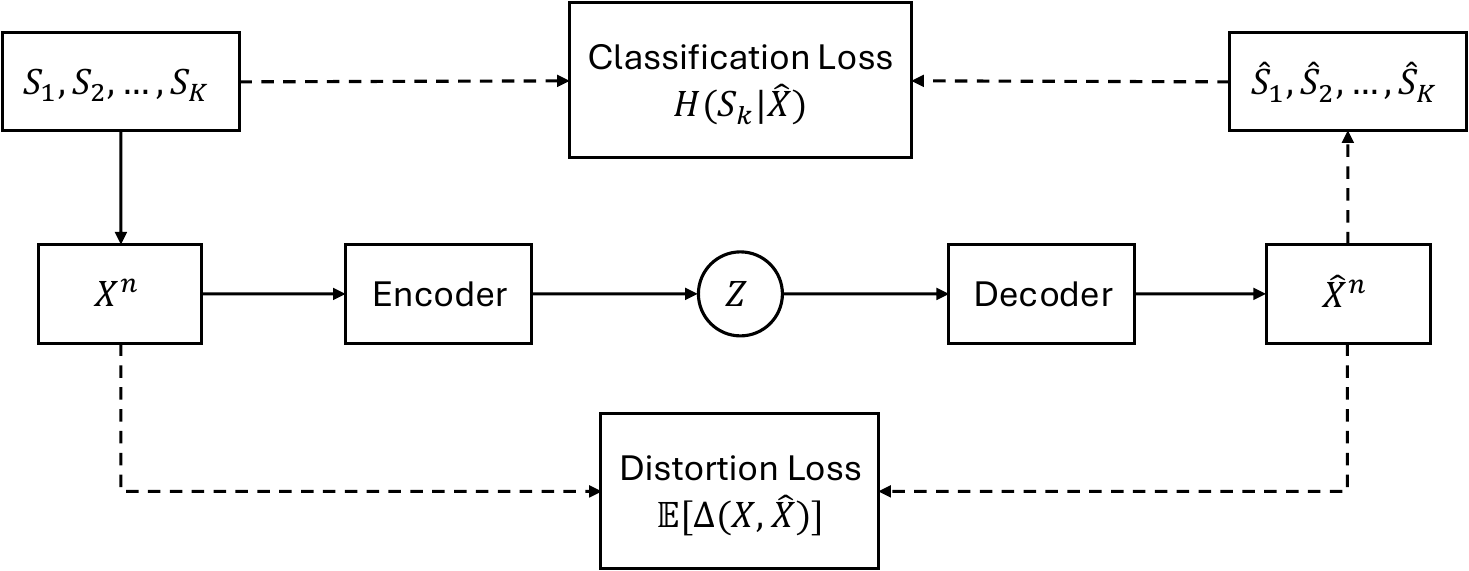}
\caption{Task-oriented lossy compression framework.}
\label{Lossy_Compression_Framework}
\end{figure}

As depicted in Figure~\ref{Lossy_Compression_Framework}, the system consists of an encoder-decoder pair.
Given a length-$n$ source sequence $X^n \sim p_X^n$, the encoder
$f: \mathcal{X}^n \to \{1,\dots,2^{nR}\}$
maps the source sequence to a finite-rate message $Z$.
The decoder
$g: \{1,\dots,2^{nR}\} \to \hat{\mathcal{X}}^n$
uses this message to produce a reconstruction $\hat{X}^n$.
The fundamental objective is to minimize the required communication rate while maintaining both reconstruction fidelity and task performance \cite{Wang2024, Nam2025}.

\textbf{Distortion constraint.}
The reconstruction $\hat{X}$ is required to satisfy an average distortion constraint
\begin{equation}
\mathbb{E}[\Delta(X, \hat{X})] \leq D,
\end{equation}
where $\Delta(\cdot,\cdot)$ denotes a prescribed distortion measure, such as Hamming distortion or mean-squared error.

\textbf{Classification constraint.}
In addition to reconstruction quality, we impose explicit constraints on the preservation of task-relevant information.
Specifically, for each target variable $S_k$, the reconstruction must satisfy
\begin{equation}
    H(S_k|\hat{X}) \leq C_k, \qquad k \in [K],
    \label{ClassificationConstraint}
\end{equation}
for some $C_k > 0$.
This constraint bounds the residual uncertainty of the task variable $S_k$ given the reconstructed source and ensures a prescribed level of classification performance \cite{Wang2024, Nam2025}.

\begin{definition}[\textbf{One-Shot \textcolor{black}{Operational} RDC Function}]
\textcolor{black}{Let $\Theta$ be a non-empty set of distortion-classification pairs $(D, C)$, where each pair specifies a joint distortion-classification objective.} 
A rate $R$ is \textit{one-shot achievable} with respect to $\Theta$ for the source variable $X$ if there exists a random seed $U$, independent of $X$, and an encoder $p_{Z|X,U}$ such that $H(Z|U) \leq R$.
For every $(D, C) \in \Theta$, a decoder $p_{\hat{X}|Z,U}$ can be constructed to meet the constraints $\mathbb{E}[\Delta(X, \hat{X})] \leq D$ and $H(S|\hat{X}) \leq C$, where the joint distribution satisfies $p_{X, Z, \hat{X}, U} = p_X p_U p_{Z|X,U} p_{\hat{X}|Z,U}$, equivalently $X \rightarrow Z \rightarrow \hat{X}$ given $U$.
The infimum of such rates is denoted by $R^*(\Theta)$. \label{def:RDC_single}
When $\Theta$ consists of a single $(D, C)$ pair:
\begin{mini!}|s|[2]
{p_U, p_{M|X,U}, p_{\hat{X}|M,U}}
{ H(M|U)}
{\label{RDPC}}
{R^*(D, C) =}
\addConstraint{\mathbb{E}[\Delta(X, \hat{X})]}{\leq D}
\addConstraint{H(S | \hat{X})}{\leq C.}
\end{mini!}
\end{definition}

\begin{figure}[]
\centering
\includegraphics[width=0.34\textwidth]{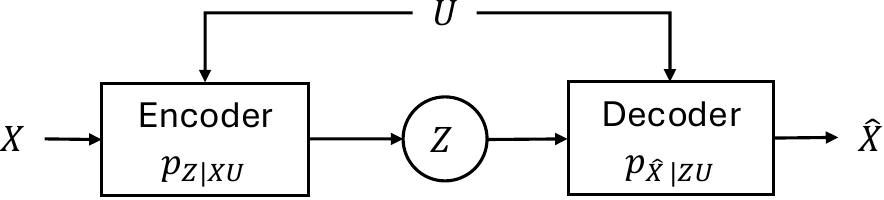}
\caption{One-shot setting with common randomness.}
\label{One_Shot_Common_Randomess}
\end{figure}

\textbf{Role of the common randomness $U$.}
In Figure \ref{One_Shot_Common_Randomess}, the random seed $U$ represents common randomness shared between the encoder and the decoder.
When such common randomness is available, the message $Z$ can be encoded using approximately $H(Z|U)$ bits through variable-length coding.

\begin{definition}[\textbf{Asymptotic \textcolor{black}{Operational} RDC Function}]
A rate $R$ is \textit{asymptotically achievable} with respect to $\Theta$ for the i.i.d.\ source sequence
$\{X^{(t)}\}_{t=1}^{\infty} \overset{\text{i.i.d.}}{\sim} p_X$
if for some positive integer $n$, there exists a random seed $U$ and an encoder $p_{Z|X^n, U}$ such that $\frac{1}{n} H(Z | U) \leq R$.
For every $(D, C) \in \Theta$, a decoder $p_{\hat{X}^n | Z, U}$ can be constructed to satisfy $\frac{1}{n} \sum_{t=1}^n \mathbb{E}[\Delta(X^{(t)}, \hat{X}^{(t)})] \leq D$ and $\frac{1}{n} \sum_{t=1}^n H(S|\hat{X}^{(t)}) \leq C$, where $p_{X^n, Z, \hat{X}^n, U}
= p_{X^n} p_U p_{Z | X^n, U} p_{\hat{X}^n | Z, U}$. The infimum of such achievable rates is denoted by $R^{(\infty)}(\Theta)$. When $\Theta$ consists of a single $(D, C)$ pair: 
\begin{mini*}|s|[2] 
{p_U, p_{Z | X^n, U}, p_{\hat{X}^n | Z, U}} 
{\frac{1}{n} H(Z | U)} 
{\label{RDPC}} 
{R^{(\infty)}(D, C) =} 
\addConstraint{\!\!\!\!\!\!\!\!\!\!\!\!\!\!\!\!\!\!\!\!\!\!\!\! \frac{1}{n} \sum_{t=1}^n \mathbb{E}[\Delta(X^{(t)}, \hat{X}^{(t)})]}{\leq D} 
\addConstraint{\!\!\!\!\!\!\!\!\!\!\!\!\!\!\!\!\!\!\!\!\!\!\!\! \frac{1}{n} \sum_{t=1}^n H(S|\hat{X}^{(t)})}{\leq C.} 
\end{mini*} 
\end{definition}

\subsection{RDC/DRC Functions for Bernoulli Sources}
We now consider the one-shot RDC framework to a binary setting. The problem can be expressed entirely in terms of the conditional distribution $p_{\hat{X}|X,U}$ with $X$ and $U$ being independent, which directly specifies the reconstruction $\hat{X}$ and obviates the need for an explicit intermediate representation $M$, as illustrated in Figure~\ref{fig:Architectures} and formalized in Theorem~\ref{thm:oneshot_random_simplified}. 

\begin{figure}[]
\centering
\includegraphics[width=0.87\linewidth]{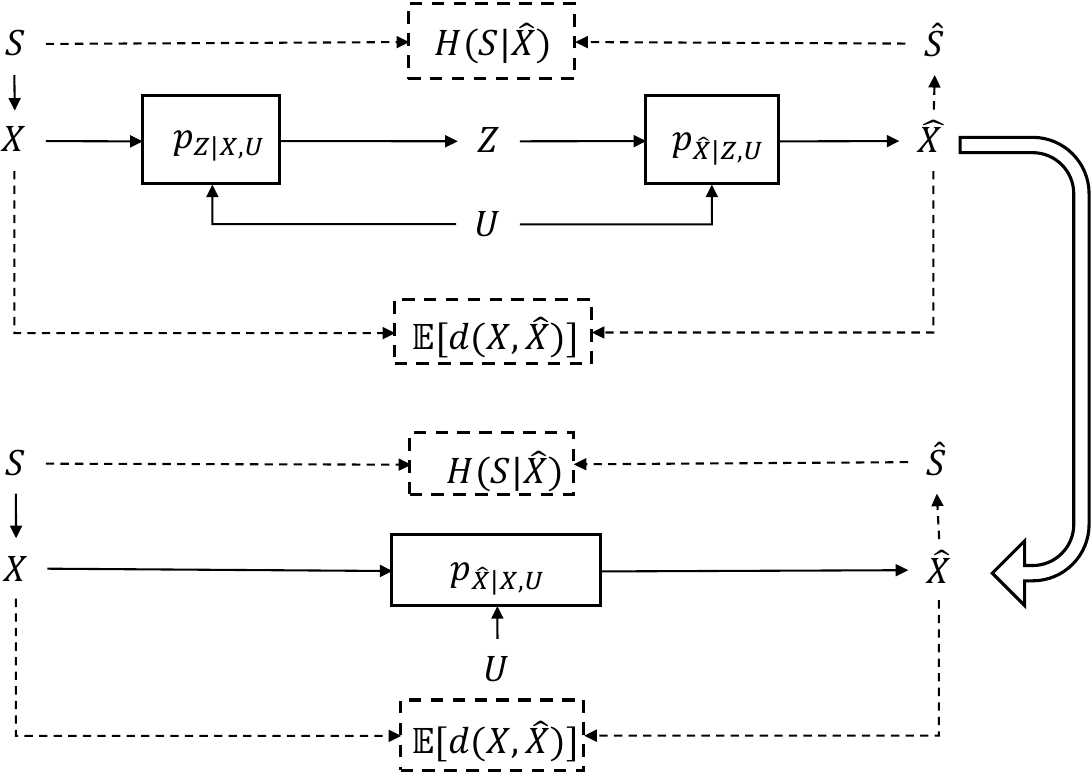}
\caption{System architecture of Theorem~\ref{thm:oneshot_random_simplified}.}
\label{fig:Architectures}
\end{figure}

\begin{theorem}
\label{thm:oneshot_random_simplified}
Suppose that $\Theta$ contains a single pair $(D, C)$.
For a source \( X \sim p_X \) and an associated classification variable \( S \sim p_S \), the one-shot RDC function can be written as
\begin{mini!}|s|[2] 
{p_U, p_{\hat{X}|X,U}} 
{ H(\hat{X}|U) }  
{\label{prob:oneshot_random_simplified}} 
{R^*(D, C) = } 
\addConstraint{\!\!\!\!\!\!\!\!\!\!\!\!\!\! H(\hat{X}|U,X) = 0, 
 I(X,U) = 0}{}  
\addConstraint{\!\!\!\!\!\!\!\!\!\!\!\!\!\! \mathbb{E}[\Delta(X, \hat{X})]\leq D, H(S | \hat{X}) \leq C.}{}  
\end{mini!}  
where $p_{U,X,\hat{X}}=p_U\,p_X\,p_{\hat{X}|U,X}$.
\end{theorem}

\begin{proof}
The claim can be shown based on Theorem~3 in~\cite{liu2022lossy}. A self-contained derivation is included in Appendix~\ref{app:proof_oneshot_random_simplified}.
\end{proof}

The constraint $H(\hat{X}|X,U)=0$ enforces that $\hat{X}$ is a deterministic function of $(X,U)$, i.e., $\hat{X}=f(X,U)$ for some mapping $f$, with the auxiliary variable $U$ serving as the sole source of randomness. In this equivalent architecture, the encoder deterministically maps $(X,U)$ to $\hat{X}$ and subsequently compresses $\hat{X}$ losslessly at an average rate approaching $H(\hat{X}|U)$, while the constraint $H(S|\hat{X})\leq C$ ensures that the relevant information for classification is preserved. The decoder then performs lossless decompression and outputs the reconstruction $\hat{X}$.

Next, consider a Bernoulli source \( X \sim \text{Bern}(q_X) \) with $0 \leq q_X \leq \frac{1}{2}$.
Let $\oplus$ denote modulo-$2$ addition so that $X \oplus \hat{X} = 1$ is equivalent to $X \neq \hat{X}$.
We assume that the task variable \( S \) is binary and is related to $X$ through a binary symmetric model,
$S = X \oplus S_1$, where \( S \sim \text{Bern}(.) \) and \( S_1 \sim \text{Bern}(q_{S_1}) \) with $0 \leq q_{S_1} \leq \frac{1}{2}$.
We adopt the Hamming distortion $\Delta_H(X,\hat{X})=\mathbf{1}\{X\neq \hat{X}\}$, under which the expected distortion can be written as $\mathbb{E}[\Delta_H(X,\hat{X})] = P[X \neq \hat{X}] 
= P(X=0,\hat{X}=1) + P(X=1,\hat{X}=0)$.

\begin{figure}[]
\centering
\includegraphics[width=0.35\textwidth]{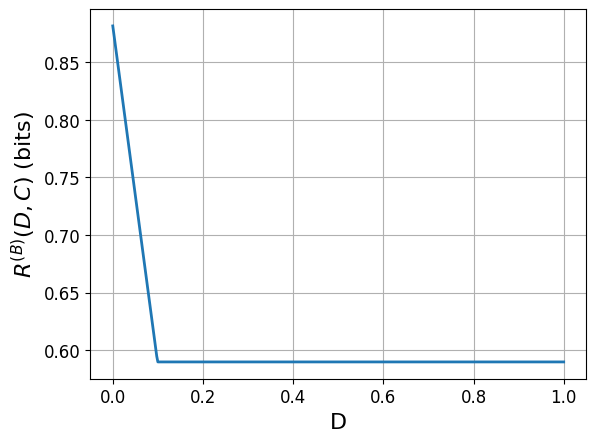}
\caption{Illustration of Theorem \ref{Oneshot_Bernoulli_radom_RDC}. $R^{(B)}(D,C)$ versus $D$ with given $C=0.8$, $q_X = 0.3$, $q_{S_1} = 0.2$.}
\label{RDC_CF}
\end{figure}

The next result gives an explicit single-letter expression for the one-shot RDC function in this Bernoulli setting when common randomness is available.

\begin{theorem}
\label{Oneshot_Bernoulli_radom_RDC}
Let $X \sim \text{Bern}(q_X)$ be a Bernoulli source and let $S$ be a binary task variable jointly distributed with $X$ via the binary symmetric model $S=X\oplus S_1$, where $S\sim \text{Bern}(q_S)$ and $S_1\sim \text{Bern}(q_{S_1})$ ($0 \leq q_X, q_S, q_{S_1}\leq\frac{1}{2}$).
Assume the Hamming distortion measure.
Then the optimization problem (\ref{prob:oneshot_random_simplified}) is feasible if $C\geq H_b(q_{S_1})$.
Moreover, under common randomness, \textcolor{black}{one-shot operational RDC is derived as}
\begin{align*}
\label{}
R_{ }^{(B)}(D, C) =  
\begin{cases} 
    \frac{H_b(q_X)(q_X - D)}{q_X}, \hfill 
    0 \leq D < \frac{q_X(C - H_b(q_{S_1})}{H_b(m) - H_b(q_{S_1})} \\
    \frac{H_b(q_X)[H_b(m) - C]}{H_b(m) - H_b(q_{S_1})} ,  
    \hfill \frac{q_X [C - H_b(q_{S_1})]}{H_b(m) - H_b(q_{S_1})} \leq D \leq 1 \\
    0, \hfill C \geq H_b\left(m\right) \text{ and } q_X < D \leq 1.
\end{cases}
\end{align*}
where $m = (1 - q_X)(1 - q_{S_1}) +  q_X q_{S_1}$ and $H_b(.)$ denotes the binary entropy function.
\end{theorem}

\begin{proof}
A complete proof is deferred to Appendix~\ref{app:proof_Oneshot_Bernoulli_radom_RDC}.
\end{proof}

To benchmark the one-shot characterization above, we next recall the \textcolor{black}{information RDC function without common randomness} for the Bernoulli model from~\cite{Wang2024}.

\begin{theorem} 
\label{RDC_Wang2024}
\cite{Wang2024} Let $X \sim \text{Bern}(q_X)$ and let $S$ be a binary task variable jointly distributed with $X$ through $S=X\oplus S_1$, where $S\sim \text{Bern}(q_S)$ and $S_1\sim \text{Bern}(q_{S_1})$ ($0 \leq q_X, q_S, q_{S_1}\leq\frac{1}{2}$).
The problem is infeasible if $C<H_b(q_{S_1})$.
Otherwise, under Hamming distortion \textcolor{black}{and without common randomness}, the \textcolor{black}{information} RDC function is
\begin{align*}
\textcolor{black}{R^{}(D,C)} = \begin{cases}
H(b) - H(D),  \hfill D< C_0 \text{ and } D\leq b,\\
H(b) - H(C_0), \hfill D\geq C_0 \text{ and } C_0\leq b,\\
0, \hfill \min\{D,C_0\}> b.
\end{cases}
\end{align*}
where $b=\min\left\{\frac{q_X(1 - 2 q_{S_1})}{1-2q_{S_1}}, 1-\frac{q_X(1 - 2 q_{S_1})}{1-2q_{S_1}}\right\} $ and $C_0= \frac{H^{-1}(C)-q_{S_1}}{1-2q_{S_1}}$. Here $H^{-1}:[0,1]\rightarrow [0,\frac{1}{2}]$ denotes the inverse function of Shannon entropy.
\end{theorem}

\begin{figure}[]
\centering
\includegraphics[width=0.35\textwidth]{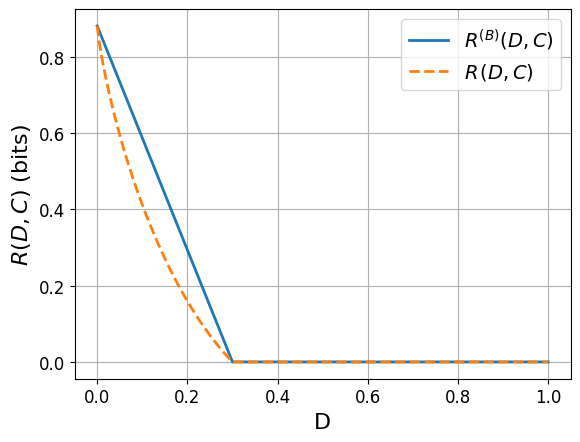}
\caption{RDC curves for a fixed \(C\): \(R^{(B)}(D,C)\) and \textcolor{black}{\(R^{}(D,C)\)} versus \(D\), for \(C=0.9\), \(q_X=0.3\), and \(q_{S_1}=0.2\).}
\label{fig:RDC_compare}
\end{figure}

Figure~\ref{fig:RDC_compare} illustrates the tradeoff curves, comparing the \textcolor{black}{information RDC function without common randomness} in \cite{Wang2024} and the one-shot RDC function in Theorem \ref{Oneshot_Bernoulli_radom_RDC} under the same bit-rate constraint. \textcolor{black}{As expected, the tradeoff curve of information RDC lies below its one-shot operational counterpart, reflecting the fundamental performance limit.}

We next turn to the dual viewpoint, in which the distortion is minimized subject to a rate constraint and a classification constraint.

\begin{definition}
\label{def:oneshot_random_simplified_DRC}
For a source \( X \sim p_X \) and an associated classification variable \( S \sim p_S \), the corresponding one-shot \textcolor{black}{operational} distortion minimization problem can be expressed as
\begin{mini!}|s|[2] 
{p_U, p_{\hat{X}|X,U}} 
{\mathbb{E}[\Delta(X, \hat{X})]}  
{\label{prob:oneshot_random_simplified_DRC}} 
{D^*(R, C) = } 
\addConstraint{\!\!\!\!\!\!\!\!\!\!\!\!\!\! H(\hat{X}|U,X) = 0, I(X,U) = 0 }{}  
\addConstraint{\!\!\!\!\!\!\!\!\!\!\!\!\!\! H(\hat{X}|U) \leq R, H(S | \hat{X}) \leq C.}{}  
\end{mini!}  
where $p_{U,X,\hat{X}}=p_U\,p_X\,p_{\hat{X}|U,X}$.
\end{definition}

\textcolor{black}{The following theorem characterizes $D^*(R, C)$ for the Bernoulli source.}

\begin{figure}[h]
\centering
\includegraphics[width=0.39\textwidth]{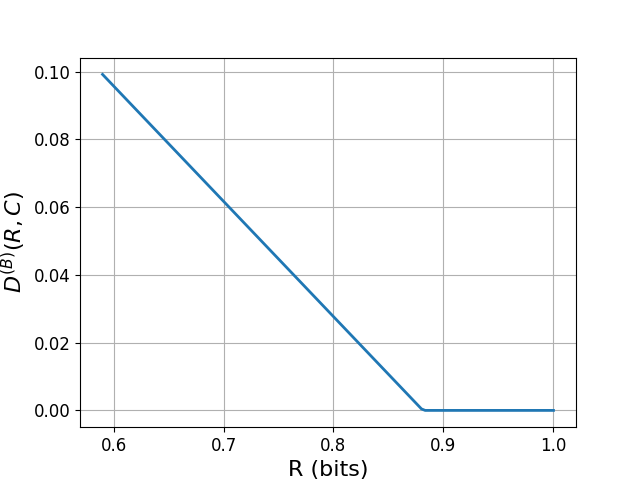}
\caption{Illustration of Theorem \ref{Oneshot_Bernoulli_radom_DRC}. $R^{(B)}(D,C)$ versus $D$ with given $C = 0.8$, $q_X = 0.3$, $q_{S_1} = 0.2$.}
\label{DRC_CF}
\end{figure}

\begin{theorem}
\label{Oneshot_Bernoulli_radom_DRC}
Let $X \sim Bern(q_X)$ be a Bernoulli source and let $S$ be a binary task variable jointly distributed with $X$ via $S=X\oplus S_1$, where $S\sim \text{Bern}(q_S)$ and $S_1\sim \text{Bern}(q_{S_1})$ ($0 \leq q_X, q_S, q_{S_1}\leq\frac{1}{2}$).
Assume the Hamming distortion measure.
Then the optimization problem (\ref{prob:oneshot_random_simplified_DRC}) is feasible if $C\geq H(q_{S_1})$.
Moreover, under common randomness, \textcolor{black}{one-shot operational DRC is derived as}
\begin{align*}
D_{ }^{(B)}(R, C) =  
\begin{cases} 
    \frac{q_X(H_b(q_X) - R)}{H_b(q_X)}, \\ 
    \hfill C > \frac{R (H_b(q_{S_1}) - H_b(m))}{H_b(q_X)}  + H_b(m) \\
    \frac{q_X(C - H_b(q_{S_1}))}{H_b(m) - H_b(q_{S_1})}, \\  
    \hfill H_b(q_{S_1}) \leq C \leq \frac{R (H_b(q_{S_1}) - H_b(m))}{H_b(q_X)}  + H_b(m) \\
    0, \hfill  C > H_b(q_S) \text{ and } R > H_b(q_X).
\end{cases}
\end{align*}
where $m = (1 - q_X)(1 - q_{S_1}) +  q_X q_{S_1}$ and $H_b(.)$ denotes the binary entropy function.
\end{theorem}

\begin{proof}
The proof is provided in Appendix~\ref{app:proof_Oneshot_Bernoulli_radom_DRC}.
\end{proof}

\section{\textcolor{black}{Information} DC Achievable Region}
Given a representation $Z$ of $X$, its \textcolor{black}{information} DC achievable region, denoted by $\Pi(p_{Z|X})$, is defined as the set of all $(D, C)$ pairs for which there exists a decoder $p_{\hat{X}|Z}$ satisfying $\mathbb{E}[\Delta(X, \hat{X})] \leq D$ and $H(S|\hat{X}) \leq C$, where the variables satisfy the Markov chain $X \rightarrow Z \rightarrow \hat{X}$.

\begin{definition}[\textbf{Lower Boundary of \( \Pi(p_{Z|X}) \)}] 
The lower boundary of DC achievable
region \( \Pi(p_{Z|X}) \), under the Hamming distortion, is defined as
\begin{mini!}|s|[2] 
{p_{\hat{X}|Z}} 
{ P(X \neq \hat{X})} 
{\label{RDPC}} 
{D(C) = } 
\addConstraint{H(S | \hat{X})}{\leq C.} 
\end{mini!}
\end{definition}

We consider a binary source $X \sim \mathrm{Bern}(q_X)$ and a (possibly multi-level) representation random variable $Z$. 
Let $Z$ take values in $[n]$ with marginal distribution $p_Z(i)=q_i$ and conditional law $p_{X|Z}(1| i)=\epsilon_i$, for $i\in[n]$, 
where $\sum_{i=1}^n q_i=1$ and $\sum_{i=1}^n q_i\,\epsilon_i=q_X$. 
These conditions induce the source marginal
$p_X(0)=\sum_{i=1}^n q_i(1-\epsilon_i)$ and $p_X(1)=1-\sum_{i=1}^n q_i(1-\epsilon_i)$, 
and ensure the Markov structure $X \rightarrow Z \rightarrow \hat{X}$ for any reconstruction $\hat{X}$ generated from $Z$.

We parameterize the decoder by letting $p_{\hat{X}|Z}(0|i) = p_i, \quad i \in [n]$. Following the approach of \cite{Qian2022}, the admissible range of \(p_i\) depends on the value of \(1-\epsilon_i\):
if \(1-\epsilon_i \geq 0.5\), then \(1-\epsilon_i \leq p_i \leq 1\); whereas if \(1-\epsilon_i < 0.5\), then \(0 \leq p_i < 1-\epsilon_i\).

\begin{theorem}
\label{theo:lowerBoundary}
Consider a Bernoulli source \(X \sim \mathrm{Bern}(q_X)\) and a classification variable \(S\) that is jointly distributed with \(X\) through a binary symmetric model \(S = X \oplus S_1\), where \(S \sim \mathrm{Bern}(q_S)\) and \(S_1 \sim \mathrm{Bern}(q_{S_1})\), with \(0 \le q_X, q_S, q_{S_1} \le \tfrac{1}{2}\).
Let \(Z\) denote a representation of \(X\), and assume that the channel between \(X\) and \(\hat{X}\) is binary symmetric. Under the Hamming distortion measure, the lower boundary of \(\Pi(p_{Z|X})\) can be characterized by solving the following linear program.
\begin{subequations}
\label{prob:D_C_lower_boundary}
\begin{align}
    D(C) &= \min_{p_1, ..., p_n} \sum_{i=1}^n \Big[ q_i (1 - \epsilon_i) + q_i (2 \epsilon_i - 1) p_i \Big] \\
    \textnormal{s.t.} \quad & (1 - 2q_{S_1}) \Big(\sum_{i=1}^{n} q_i \epsilon_i p_i \Big) \nonumber \\
    &- (H^{-1}(C) - q_{S_1}) \Big(1 - \sum_{i=1}^{n} q_i (1 - \epsilon_i \Big) \leq 0, \\
    & 1 - \epsilon_i \leq p_i \leq 1 \quad \text{if} \quad 1 - \epsilon_i \geq 0.5, \\
    &  0 \leq p_i < 1 - \epsilon_i \quad \text{if} \quad 1 - \epsilon_i < 0.5. 
\end{align}
\end{subequations}
\end{theorem}

\begin{proof}
See Appendix~\ref{app:proof_lowerBoundary} for the proof.
\end{proof}

\textcolor{black}{The problem in Theorem \ref{theo:lowerBoundary}} can be efficiently solved using standard convex optimization tools, such as CVX for MATLAB~\cite{cvx,grant_boyd_2008_graph_impl} or CVXPY for Python~\cite{diamond2016cvxpy,agrawal2018rewriting}. The lower boundary of $\Pi(p_{Z|X})$ is illustrated in Fig.~\ref{fig:DC_lower_boundary}.

\begin{figure}[]
    \centering
    \subfloat[$n = 2, q_1 = q_2 = 0.5, \epsilon_1 = 0.2, \epsilon_2 = 0.8, q_{S_1} = 0.05$]{
        \includegraphics[width=0.75\linewidth]{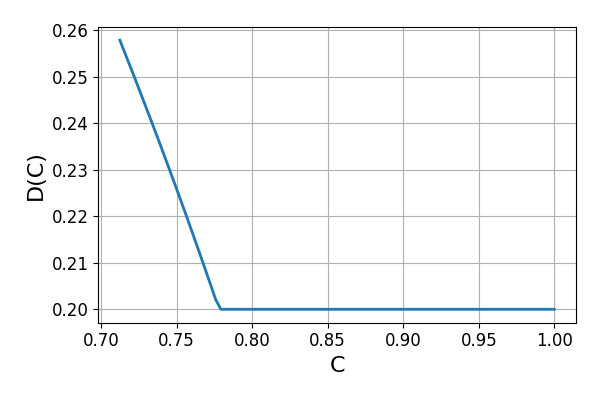}
        \label{fig:DC_Lower_Bound_1}
    }
    \hfill
    \subfloat[$n = 4, q_1 = 0.2, q_2 = 0.3, q_3 = 0.1, q_4 = 0.4,$
    $\epsilon_1 = 0.15, \epsilon_2 = 0.35, \epsilon_3 = 0.65, \epsilon_4 = 0.85, q_{S_1} = 0.1$]{
        \includegraphics[width=0.75\linewidth]{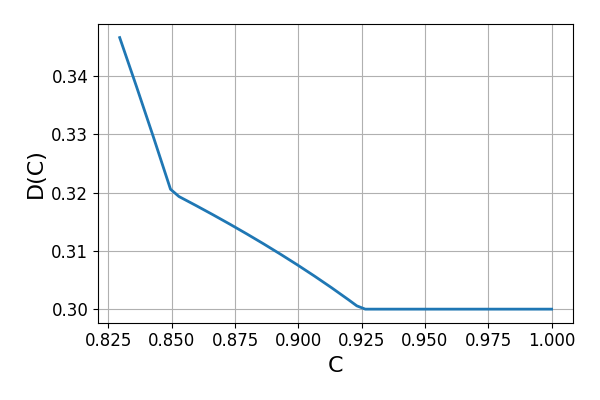}
        \label{fig:DC_Lower_Bound_2}
    }
    \caption{Lower boundary of achievable region,  $\Pi(p_{Z|X})$.}
    \label{fig:DC_lower_boundary}
\end{figure}

\begin{figure}[]
\centering
\includegraphics[width=0.4\textwidth]{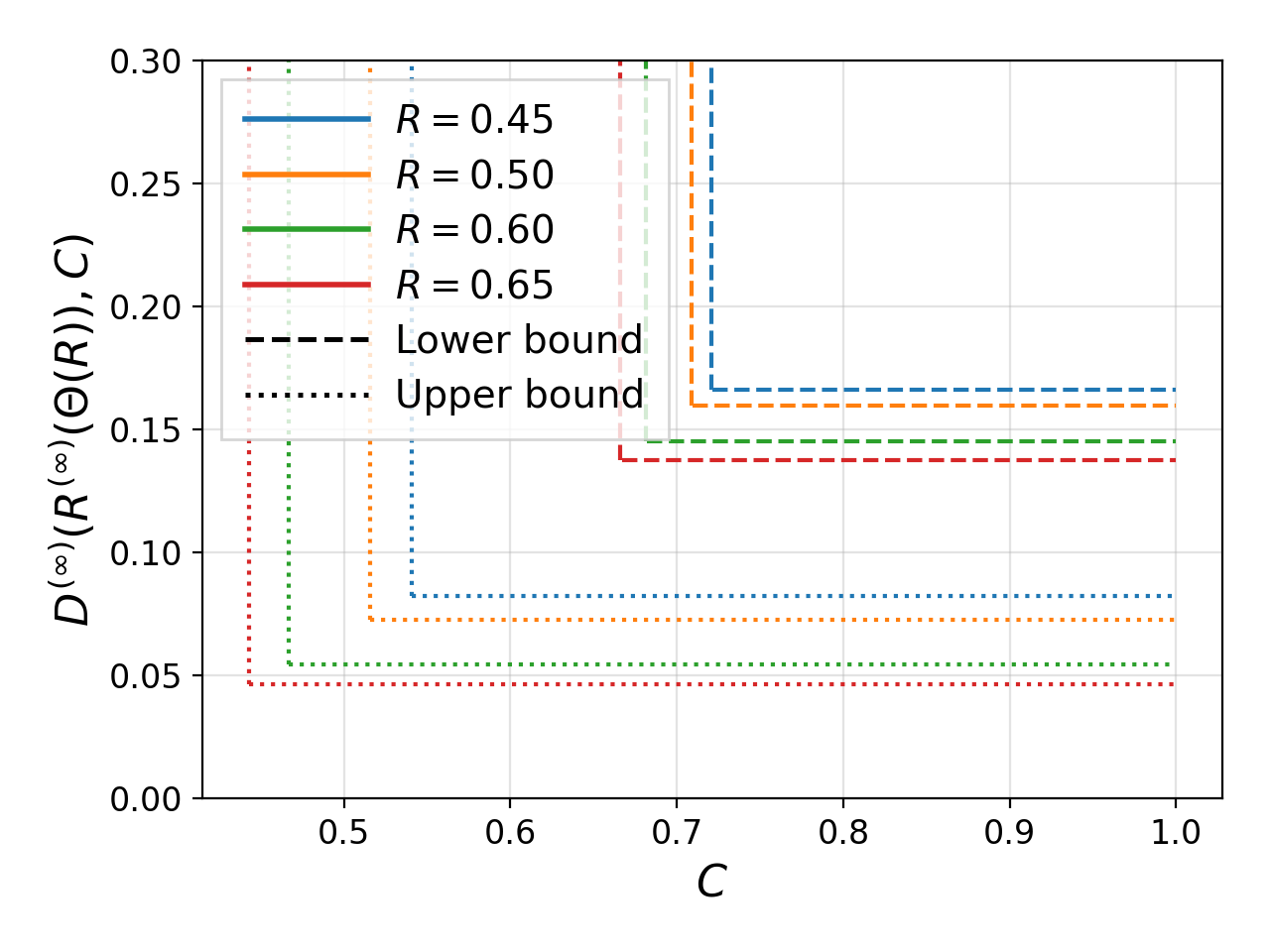}
\caption{DRC curves for a fixed \(R\): \(D^{(\infty)}\!\bigl(R^{(\infty)}_{LB}(\Theta(R)), C\bigr)\) versus \(C\) with $q_X = 0.2$, $q_{S_1} = 0.05$.}
\label{DCR_R_LB_UB}
\end{figure}

\section{Universal Encoder Representation}
We now introduce a representation-based formulation that decouples the encoder design from the specific DC operating point. This perspective allows a single encoder to support a family of objectives by pairing it with different decoders, as shown in Figure \ref{fig:Universal_Scheme} \cite{Nam2025}.

\begin{figure}[!htbp]
\centering
\includegraphics[width=0.45\textwidth]{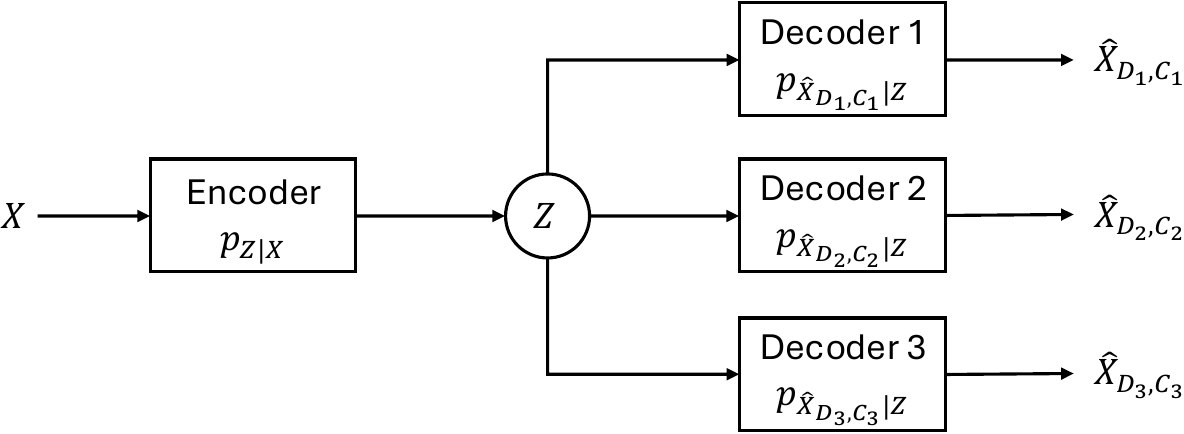}
\caption{The universal encoder representation framework.}
\label{fig:Universal_Scheme}
\end{figure}

\begin{definition}[\textbf{Universal \textcolor{black}{Information} RDC Function}]
\cite{Nam2025} Consider $Z$ be a representation of $X$ by the conditional distribution $p_{Z|X}$. Let $\mathcal{P}_{Z|X}(\Theta)$ denote the set of conditional distributions $p_{Z|X}$ such that, for every $(D, C) \in \Theta$, there exists a conditional distribution $p_{\hat{X}|Z}$ satisfying $\mathbb{E}[\Delta(X, \hat{X})] \leq D$ and $H(S|\hat{X}) \leq C$, where the joint distribution factors as $p_{X, Z, \hat{X}} = p_{X} p_{Z|X} p_{\hat{X}|Z}$, equivalently forming the Markov chain $X \rightarrow Z \rightarrow \hat{X}$. Define
\[
R(\Theta) = \inf_{p_{Z|X} \in \mathcal{P}_{Z|X}(\Theta)} I(X; Z).
\]
\end{definition}

This definition characterizes the minimum information rate required by a representation $Z$ that is sufficiently expressive to support all distortion-classification objectives in $\Theta$ through appropriate decoding. Given a representation $Z$ of $X$ with distortion-classification achievable region $\Pi(p_{Z|X})$, we say that $Z$ is a $\Theta$-universal representation of $X$ if $\Theta \subseteq \Pi(p_{Z|X})$.

\textcolor{black}{
The equality $R^{(\infty)}(\Theta) = R(\Theta)$ follows directly from Theorem~7 in~\cite{Wang2024}. In particular, the asymptotic source coding rate is fully characterized by $R(\Theta)$, which is defined through an optimization over auxiliary random variables $Z$ subject to the prescribed constraints. Accordingly, any random variable $Z$ jointly distributed with $X$ can be interpreted as a valid representation (or reconstruction) of the source.}

Under this representation-centric viewpoint, $R^{(\infty)}(\Theta)$ can be interpreted as the \textit{minimum rate} required by a single, fixed encoder to accommodate all DC requirements in $\Theta$.
Equivalently, it coincides with the infimum of $I(X;Z)$ over all $\Theta$-universal representations $Z$ of $X$.

\begin{definition}[\textbf{Rate Penalty}]
Let $\sup_{(D, C) \in \Theta} R^{(\infty)}(D, C)$ denote the rate required to satisfy the most demanding objective in $\Theta$.
The \emph{rate penalty} associated with universality is defined as
\[
\Delta(\Theta) = R^{(\infty)}(\Theta) - \sup_{(D, C) \in \Theta} R^{(\infty)}(D, C),
\]
which quantifies the additional rate needed to meet all objectives in $\Theta$ using a single encoder.
\end{definition}

The quantity $\Delta(\Theta)$ represents the minimum rate penalty incurred when using a $\Theta$-universal representation rather than tailoring an optimal representation to each individual DC pair.
A particularly important case arises when $\Theta = \Theta(R)$, where $\Theta(R) = \{(D, C): R^{(\infty)}(D, C) \leq R\}$. In this setting, $\Delta(\Theta(R))$ captures the fundamental cost of universality under a prescribed rate budget.

From Theorem \ref{RDC_Wang2024}, for a given rate $R$, define the two boundary parameters $C_0=C_0(R)$ and $C_{\min}=C_{\min}(R)$ as the solutions to
\begin{align}
\label{function_R}
\left\{\begin{matrix}
R =  H(b) - H(C_0), \\
R = H(b) - H( \frac{H^{-1}(C_{\min})- q_{S_1}}{1-2 q_{S_1}}),
\end{matrix}\right. 
\end{align}

Since \( \Theta(R) \) is achievable and \( (C_0, C = H(C_0 (1 - 2 q_{s_1})) + q_{s_1}) \) and \( (b, C_{\min}) \) lie on the boundary of this region, these points are achievable. The minimum rate \( R^{(\infty)}(\Theta(R)) \) is the solution of this problem:
\begin{mini*}|s|[2] 
{p_{\hat{X}_1, \hat{X}_2 | X}} 
{I(X; \hat{X}_1, \hat{X}_2)} 
{\label{RDPC}} 
{R^{(\infty)}(\Theta(R)) =} 
\addConstraint{\!\!\!\!\!\!\!\!\!\!\! P(X \neq \hat{X}_1)}{\leq C_0} 
\addConstraint{\!\!\!\!\!\!\!\!\!\!\! P(X \neq \hat{X}_2)}{\leq b} 
\addConstraint{\!\!\!\!\!\!\!\!\!\!\! H(S | \hat{X}_1)}{\leq H(C_0 (1 - 2 q_{s_1})) + q_{s_1})} 
\addConstraint{\!\!\!\!\!\!\!\!\!\!\! H(S | \hat{X}_2)}{\leq C_{\min}.} 
\end{mini*}

\begin{theorem}
\label{theo:ratePenaltyBound}
Assume the channel of $X$ and $\hat{X}_1$ is a binary symmetric channel. Under the Hamming distortion, the lower and upper bounds of \( R^{(\infty)}(\Theta(R)) \) can be derived by solving these following linear programming problems.
\begin{subequations}
\label{R_LB}
\begin{align}
    R^{(\infty)}_{LB}(\Theta(R)) &= \min_{p_{ij|k}} \quad I_{LB}(X; \hat{X}_1, \hat{X}_2) \\
    \textnormal{s.t.} \quad \!\!\! &(1 - q)(1 - p_{00|0} - p_{01|0}) \nonumber \\ 
    &+ q(p_{00|1} + p_{01|1}) \leq C_0, \\
    & (1 - q)(p_{01|0} + p_{11|0}) \nonumber \\
    &+ q(1 - p_{01|1} - p_{11|1}) \leq b, \\
    &  (1 - q_{s_1})(p_{10|0} + p_{11|0}) + q_{s_1} (1 - p_{10|0} - p_{11|0}) \nonumber \\
    & \leq C_0 (1 - 2 q_{s_1})) + q_{s_1}, \\
    & (1 - q_{s_1})(p_{01|0} + p_{11|0}) \nonumber \\
    &+ q_{s_1} (1 - p_{01|0} - p_{11|0}) \leq H^{-1}(C_{\min}), \\
    & p_{00|0} + p_{01|0} + p_{11|0} \leq 1, \\
    & p_{00|1} + p_{01|1} + p_{11|1} \leq 1, \\
    & 0 \leq p_{00|0}, 0 \leq p_{01|0}, 0 \leq p_{11|0} \leq 1, \\
    & 0 \leq p_{00|1}, 0 \leq p_{01|1}, 0 \leq p_{11|1} \leq 1.
\end{align}
\end{subequations}
And, 
\begin{subequations}
\label{R_UB}
\begin{align}
    R^{(\infty)}_{UB}(\Theta(R)) &= \min_{p_{ij|k}} \quad I_{LB}(X; \hat{X}_1, \hat{X}_2) \\
    \textnormal{s.t.} \quad \!\!\! &(1 - q)(1 - p_{00|0} - p_{01|0}) \nonumber \\ 
    &+ q(p_{00|1} + p_{01|1}) \leq C_0, \\
    & (1 - q)(p_{01|0} + p_{11|0}) \nonumber \\
    &+ q(1 - p_{01|1} - p_{11|1}) \leq b, \\
    &  (1 - q_{s_1})(p_{10|0} + p_{11|0}) + q_{s_1} (1 - p_{10|0} - p_{11|0}) \nonumber \\
    & \leq 0, \\
    & (1 - q_{s_1})(p_{01|0} + p_{11|0}) \nonumber \\
    &+ q_{s_1} (1 - p_{01|0} - p_{11|0}) \leq H^{-1}(C_{\min}), \\
    & p_{00|0} + p_{01|0} + p_{11|0} \leq 1, \\
    & p_{00|1} + p_{01|1} + p_{11|1} \leq 1, \\
    & 0 \leq p_{00|0}, 0 \leq p_{01|0}, 0 \leq p_{11|0} \leq 1, \\
    & 0 \leq p_{00|1}, 0 \leq p_{01|1}, 0 \leq p_{11|1} \leq 1.
\end{align}
\end{subequations}
where \(p_{ij|k}\), for \(i,j,k \in \{0,1\}\), denotes the conditional pmf \(p_{\hat{X}_1,\hat{X}_2 | X}(i,j | k)\).
\end{theorem}

\begin{proof}
The proof is provided in Appendix~\ref{app:proof_ratePenaltyBound}.
\end{proof}

For a fixed rate $R$, compute $C_0$ and $C_{\min}$ according to~\eqref{function_R}. 
Next, solve the linear programs~\eqref{R_LB} and~\eqref{R_UB} to obtain 
$R^{(\infty)}_{LB}(\Theta(R))$ and $R^{(\infty)}_{UB}(\Theta(R))$, respectively. 
This allows us to plot the curves 
(i) $D^{(\infty)}(R^{(\infty)}_{LB}(\Theta(R)), C)$ versus $C$ and 
(ii) $D^{(\infty)}(R^{(\infty)}_{UB}(\Theta(R)), C)$ versus $C$, 
as shown in Fig.~\ref{DCR_R_LB_UB}.

\section{Conclusion}
We studied task-oriented lossy compression through the framework of RDC representations for Bernoulli sources under Hamming distortion. In the one-shot setting with common randomness, we derived closed-form characterizations of the RDC/DRC tradeoffs. With a representation-based viewpoint, we characterized the achievable DC region induced by a fixed representation via a linear-programming formulation. We further investigated universal encoders that can support a family of DC operating points and established computable lower and upper bounds on the minimum rate penalty required for universality across multiple task requirements. Numerical examples illustrate the achievable regions and the resulting universal RDC/DRC curves, providing an analytically tractable benchmark for understanding universality in task-oriented lossy compression.

\clearpage
\newpage
\bibliographystyle{unsrt}
\bibliography{main}

\clearpage
\newpage
\appendix

\subsection{Proof of Theorem \ref{thm:oneshot_random_simplified}}
\label{app:proof_oneshot_random_simplified}

We start from the one-shot RDC formulation with common randomness given in
Definition~\ref{def:RDC_single}:
\begin{mini*}|s|[2] 
{p_U, p_{Z|X,U}, p_{\hat{X}|Z,U}} 
{ H(Z|U) } 
{} 
{R^*(D, C) = } 
\addConstraint{\mathbb{E}[\Delta(X, \hat{X})]}{\leq D}  
\addConstraint{H(S | \hat{X})}{\leq C.}  
\end{mini*}  
where $p_{U,X,Z,\hat{X}}=p_U\,p_X\,p_{Z|X,U}\,p_{\hat{X}|Z,U}$.

\noindent\textbf{Upper bound.}
Fix any joint distribution $p_{U,X,\hat{X}}$ satisfying
$\mathbb{E}[\Delta(X, \hat{X})] \le D$ and $H(S|\hat{X})\le C$, and define
$Z\triangleq \hat{X}$.
With this choice, the induced joint distribution becomes
$p_{U,X,Z,\hat{X}}=p_U\,p_X\,\delta_{Z=\hat{X}(X,U)}\,\delta_{\hat{X}=Z}$.
Moreover, $H(M|U)=H(\hat{X}|U)$, while both
$\mathbb{E}[\Delta(X,\hat{X})]$ and $H(S|\hat{X})$ remain unchanged.
Consequently, we obtain
\begin{mini!}|s|[2] 
{p_U, p_{\hat{X}|X,U}} 
{ H(\hat{X}|U) } 
{} 
{R^*(D, C) \leq } 
\addConstraint{H(\hat{X}|U,X)}{= 0} 
\addConstraint{I(X,U)}{= 0} 
\addConstraint{\mathbb{E}[\Delta(X, \hat{X})]}{\leq D} 
\addConstraint{H(S | \hat{X})}{\leq C.} 
\end{mini!}

\noindent\textbf{Tightness.}
Conversely, consider any feasible joint distribution
$p_{U,X,Z,\hat{X}}$ satisfying
$\mathbb{E}[\Delta(X, \hat{X})] \le D$ and $H(S|\hat{X})\le C$.
By the functional representation lemma
\cite{elgamal2011network, li2018strong}, there exist:
(i) a random seed $V_1$, independent of $(U,X)$, and a measurable mapping
$\phi_1$ such that $Z=\phi_1(U,X,V_1)$ in distribution for $p_{Z|X,U}$;
(ii) a random seed $V_2$, independent of $(U,X,V_1)$, and a measurable mapping
$\phi_2$ such that $\hat{X}=\phi_2(U,Z,V_2)$ in distribution for
$p_{\hat{X}|Z,U}$.

Define the augmented random seed $U' \triangleq (U,V_1,V_2)$.
Then
$\hat{X}=\phi_2\big(U,\,\phi_1(U,X,V_1),\,V_2\big)$
is deterministic given $(U',X)$, implying $H(\hat{X}|U',X)=0$.
Furthermore, the induced joint distribution satisfies
$p_{U',X,Z,\hat{X}}=p_{U'}\,p_X\,p_{Z|X,U'}\,p_{\hat{X}|Z,U'}$.
Since the marginal distribution of $(X,\hat{X})$ is preserved,
both $\mathbb{E}[\Delta(X,\hat{X})]$ and $H(S|\hat{X})$ remain unchanged.

For the rate term, conditioning reduces entropy, and the deterministic
structure yields
\[
H(Z|U)\ \ge\ H(Z|U,V_1,V_2)=H(Z|U')\ \ge\ H(\hat{X}|U'),
\]
which leads to
\begin{mini!}|s|[2] 
{p_{U'}, p_{\hat{X}|X,U'}} 
{ H(\hat{X}|U') } 
{} 
{R^*(D, C) \geq } 
\addConstraint{H(\hat{X}|U',X)}{= 0} 
\addConstraint{I(X,U')}{= 0} 
\addConstraint{\mathbb{E}[\Delta(X, \hat{X})]}{\leq D} 
\addConstraint{H(S | \hat{X})}{\leq C.} 
\end{mini!}

Since the auxiliary alphabet is unrestricted, the random seed $U'$ can be relabelled as $U$ within the minimization without loss of generality.
Combining the upper and lower bounds completes the proof.

\subsection{Proof of Theorem \ref{Oneshot_Bernoulli_radom_RDC}}\label{app:proof_Oneshot_Bernoulli_radom_RDC}

We study (\ref{prob:oneshot_random_simplified}) with the Hamming distortion setting. In this case, the one-shot Bernoulli RDC function can be written as
\begin{subequations}
\begin{align}
    R_{ }^{(B)}(D, C) &= \inf_{p_U, p_{\hat{X}|X,U}} I(X;\hat{X}|U) \\
    \textnormal{s.t.} \quad & H(\hat{X}|U,X) = 0, \\
    & I(X,U) = 0, \\
    &  P(X \neq \hat{X}) \leq D, \\
    & H(S|\hat{X}) \leq C.
\end{align}
\end{subequations}

Note that $ H(\hat{X}|U) = I(X;\hat{X}|U) + H(\hat{X}|U,X) = I(X;\hat{X}|U) $. 
Moreover, $H(\hat{X}|U,X) = 0$ implies that $\hat{X}$ is a deterministic function of $(X,U)$, i.e., $\hat{X} = f(X,U)$ for some mapping $f$. Therefore, the optimization may be viewed as selecting a distribution $p_U$ over such deterministic mappings:
\begin{subequations}
\begin{align}
    R_{ }^{(B)}(D, C) &= \inf_{p_U} I(X;\hat{X}|U) \\
    \textnormal{s.t.} \quad & H(\hat{X}|U,X) = 0, \\
    & I(X,U) = 0, \\
    &  P(X \neq \hat{X}) \leq D, \\
    & H(S|\hat{X}) \leq C.
\end{align}   
\end{subequations}

Because Shannon entropy is defined only for discrete random variables, the auxiliary variable $U$ is chosen so that, for every realization $u$, the conditional distribution of $\hat{X}$ given $U=u$ is discrete (this requirement persists even when $(X,\hat{X})$ are otherwise modeled in a continuous space) \cite{liu2022lossy}. 

Accordingly, we may assume that $p_U$ is supported on the finite index set $\mathcal{U} \triangleq \{1,2,\ldots,|\mathcal{\hat{X}}|^{|\mathcal{X}|}\}$, where each $u \in \mathcal{U}$ corresponds to a deterministic mapping $f_u:\mathcal{X}\to\mathcal{\hat{X}}$, and the collection $\{f_u : u \in \mathcal{U}\}$ enumerates all distinct such mappings. By the support lemma (Appendix~C, p.~631 of \cite{elgamal2011network}), it suffices to assign positive probability to at most $|\mathcal{\hat{X}}|+1$ of these mappings.

In the finite-alphabet setting, this yields the equivalent finite-dimensional program
\begin{subequations}
\begin{align}
   R_{ }^{(B)}(D, C) &= \min_{p_U}  \sum_{u \in \mathcal{U}} p_U(u) I(X;\hat{X}|U = u)\\
    \textnormal{s.t.} \quad  & \sum_{u \in \mathcal{U}} p_U(u) \, P(X \neq \hat{X}| U = u)  \leq D, \\
    &\sum_{u \in \mathcal{U}} p_U(u) H(S|f_u(X)) \leq C.
\end{align}    
\end{subequations}

For a binary source with a binary reconstruction alphabet, the auxiliary alphabet can be restricted to at most four symbols without loss of optimality. This follows from the fact that there are exactly four distinct deterministic mappings from $\{0,1\}$ to $\{0,1\}$, namely
$f_1(x)=x$, $f_2(x)=1-x$, $f_3(x)=0$, and $f_4(x)=1$, for $x\in\{0,1\}$. As a result, the reconstruction $\hat{X}$ can be expressed as
\[
\hat{X} =
\begin{cases}
X, & U=1,\\
1-X, & U=2,\\
0, & U=3,\\
1, & U=4.
\end{cases}
\]

From this representation, we compute
\begin{align*}
I(X;\hat{X}|U) &= \sum_{u \in \mathcal{U}} p_U(u) I(X;\hat{X}|U = u) \\
&= \sum_{u \in \mathcal{U}} p_U(u) H(f_u(X)) \\
&= H_b(q_X)(p_U(1) + p_U(2)), \\
P(X \neq \hat{X}) &= \sum_{u \in \mathcal{U}} p_U(u) P(X \neq \hat{X}| U = u) \\
&= p_U(2) + q_X p_U(3) + (1 - q_X)p_U(4).
\end{align*}

Next, since $S \rightarrow X \rightarrow \hat{X}$, the data-processing inequality \cite{cover1999elements} implies
\begin{align*}
H(S|\hat{X}) \geq H(S|X) = H(X\oplus S_1|X) = H(S_1).
\end{align*}
Hence, feasibility of the classification constraint requires $C\geq H(S_1)$.

We now evaluate $H(S|\hat{X},U=u)$ for each mapping.  

For \( U = 1 \): \( \hat{X} = X \)
    \begin{align*}
    H(S | \hat{X}, U = 1) &= H(S | X) \\
    &= H(X \oplus S_1 | X) = H(S_1) = H_b(q_{S_1}).   
    \end{align*}

For \( U = 2 \): \( \hat{X} = 1 - X \)
    \begin{align*}
    H(S | \hat{X}, U = 2) = H(S | X) = H(S_1) = H_b(q_{S_1}).    
    \end{align*}

For \( U = 3 \): \( \hat{X} = 0 \)
    \begin{align*}
    S = X \oplus S_1  &\Rightarrow  P(S = 0) = (1 - q_X)(1 - q_{S_1}) + q_X q_{S_1}, \\
    H(S | \hat{X}, U = 3) &= H(S | U = 3) \\
    &= H_b((1 - q_X)(1 - q_{S_1}) + q_X q_{S_1}).
    \end{align*}

For \( U = 4 \): \( \hat{X} = 1 \)
    \begin{align*}
    H(S | \hat{X}, U = 4) &= H(S | U = 4) \\
    &= H_b((1 - q_X)(1 - q_{S_1}) + q_X q_{S_1}).    
    \end{align*}

Therefore, averaging over $U$ gives
\begin{align*}
H(S | \hat{X}) &= \sum_{u \in \mathcal{U}} p_U(u) H(S|f_u(X)) \\
&=  (p_U(1) + p_U(2)) H_b(q_{S_1}) \\
&+ (p_U(3) + p_U(4)) H_b((1 - q_X)(1 - q_{S_1}) +  q_X q_{S_1}),
\end{align*}
and with $m = (1 - q_X)(1 - q_{S_1}) +  q_X q_{S_1}$, this becomes
\begin{align*}
H(S | \hat{X}) &= (p_U(1) + p_U(2)) H_b(q_{S_1}) \\
&+ (p_U(3) + p_U(4)) H_b(m).
\end{align*}

Substituting the above expressions into the constraints yields the linear program
\begin{subequations}
\begin{align}
R_{ }^{(B)}(D, C) 
&= \!\!\!\!\!\!\!\!\!\!\! \min_{p_U(1), p_U(2), p_U(3),  p_U(4)} \!\!\!\!\!\!\!\!\!\!\! H_b(q_X)(p_U(1) + p_U(2))
\label{DRC_Bernoulli_random}\\
\textnormal{s.t.} \quad
& \! p_U(2) + q_X p_U(3) + (1 - q_X)p_U(4) \leq D, \label{rate_common_constraint_1_DRC} \\
& \! (p_U(1) + p_U(2)) H_b(q_{S_1}) \nonumber \\
& \! + (p_U(3) + p_U(4)) H_b(m) \leq C \label{rate_common_constraint_3_DRC}\\
&\! p_U(1) + p_U(2) + p_U(3) + p_U(4) = 1, \label{rate_common_constraint_4_DRC}\\
&\! p_U(1),\, p_U(2),\, p_U(3),\, p_U(4) \geq 0 \label{rate_common_constraint_5_DRC}.
\end{align}    
\end{subequations}

\textbf{Structure and activity of the nonnegativity constraints.}
Let $a\triangleq p_U(1)+ p_U(2)$ so that $ p_U(3)+ p_U(4)=1-a$. Write $ p_U(3)=t\in[0,1-a]$, $ p_U(4)=(1-a)-t$. Then
\begin{align*}
\mathbb{E}[d_H(X,\hat{X})] &= p_U(2)+q_X  p_U(3)+(1-q_X)  p_U(4) \\
&=  p_U(2) + (2q_X-1)t + (1-q_X)(1-a).    
\end{align*}

Since $0\le q_X\le \tfrac12$, we have $2q_X-1\le 0$. Thus, for any fixed $a$, the distortion is minimized by choosing $t=1-a$ (equivalently, $p_U(3)=1-a$ and $p_U(4)=0$). It follows that
\begin{align*}
&\min_{ p_U(3)+  p_U(4)=1-a} \bigl[ p_U(2) +q_X  p_U(3)+(1-q_X)  p_U(4)\bigr] \\
& = \min_{1-a} p_U(2) + q_X(1-a).   
\end{align*}

Moreover, the objective depends on $(p_U(1)+p_U(2))$ and is increasing in $p_U(2)$, while $p_U(2)$ also increases the distortion. Hence, an optimal solution must satisfy $ p_U(2)^\star=0$.
Therefore, the problem reduces to a single scalar variable $a=p_U(1)\in[0,1)$:
\begin{mini!}|s|[2] 
{a} 
{ H_b(q_X)a } 
{\label{prob:DRC_Bernoulli_random_one_var}} 
{R_{ }^{(B)}(D, C) = } 
\addConstraint{\!\! q_X(1-a)\le D}{\label{eq:dist-reduced-qx}} 
\addConstraint{\!\! aH_b(q_{S_1}) + (1-a)H_b(m) \le C}{\label{eq:class-reduced}} 
\addConstraint{\!\! a \geq 0.}{} 
\end{mini!}
and the corresponding optimizer in the original variables is $(p_U^\star(1),p_U^\star(2),p_U^\star(3),p_U^\star(4))=(a^\star,\,0,\,1-a^\star,\,0)$.

\textbf{Auxiliary fact used in the case analysis.}
We next record a simple monotonicity property that will be used repeatedly.

\begin{lemma}\label{lem:entropy_order}
Let $m=(1-q_X)(1-q_{S_1})+q_X q_{S_1}$. Then
\[
H_b(m)\;\ge\; H_b(q_{S_1}),
\]
with equality if only if $q_X\in\{0,1\}$ or $q_{S_1}=\tfrac12$.
\end{lemma}

\begin{proof}[Proof of Lemma \ref{lem:entropy_order}]
The identity for $m$ is immediate:
\begin{align*}
&m = (1-q_X)(1-q_{S_1}) + q_X q_{S_1} = \tfrac{1}{2} + \big(q_X - \tfrac{1}{2}\big)\big(2q_{S_1}-1\big),\\
&m-\tfrac12 \; =\; \bigl(q_X-\tfrac12\bigr)\bigl(2q_{S_1}-1\bigr).
\end{align*}

Taking absolute values and using $|2q_{S_1}-1|=2|q_{S_1}-\tfrac12|$ gives
\[
\bigl|m-\tfrac12\bigr|
\;=\; 2\,\bigl|q_X-\tfrac12\bigr|\,\bigl|q_{S_1}-\tfrac12\bigr|
\;\le\; \bigl|q_{S_1}-\tfrac12\bigr|
\]
because $|q_X-\tfrac12|\le \tfrac12$ for $q_X\in[0,1]$. The binary entropy is
maximized at $\tfrac12$ and strictly decreases with $|p-\tfrac12|$, hence
$H_b(m)\ge H_b(q_{S_1})$, with equality iff $|q_X-\tfrac12|=\tfrac12$ (i.e., $q_X\in\{0,1\}$) or
$|q_{S_1}-\tfrac12|=0$ (i.e., $q_{S_1}=\tfrac12$).
\end{proof}

We will invoke Lemma~\ref{lem:entropy_order} 
in the analysis below to justify the sign of denominators of the form $H_b(m)-H_b(q_{S_1})$.

Since (\ref{prob:DRC_Bernoulli_random_one_var}) is a linear programming problem, its optimum is attained at an extreme point and may be characterized by examining which constraints are active. Below we perform a complete case analysis by considering all possible activity patterns of (\ref{eq:dist-reduced-qx}) and (\ref{eq:class-reduced}).

\textbf{Case 1.} Constraint~(\ref{eq:dist-reduced-qx}) is active and constraint~(\ref{eq:class-reduced}) is inactive.

The constraint~(\ref{eq:dist-reduced-qx}) is active, we get
\begin{align*}
q_X(1-a) = D \Rightarrow a = \frac{q_X - D}{q_X}.
\end{align*}

As a consequence, we have
\[
R_{ }^{(B)}(D, C) =  \frac{H_b(q_X)(q_X - D)}{q_X}.
\]

And,
\begin{align*}
p_U^\star(1) &= \frac{q_X - D}{q_X},\\
p_U^\star(2) &= 0,\\
p_U^\star(3) &= \frac{D}{q_X},\\
p_U^\star(4) &= 0.
\end{align*}

The constraint~(\ref{eq:class-reduced}) is inactive if
\begin{align*}
&a\,H_b(q_{S_1}) + (1-a)H_b(m) < C, \\
&\frac{q_X - D}{q_X} \,H_b(q_{S_1}) + \left(1-\frac{q_X - D}{q_X} \right)H_b(m) < C, \\
& D < \frac{q_X(C - H_b(q_{S_1})}{H_b(m) - H_b(q_{S_1})}.
\end{align*}

Therefore, $
R_{ }^{(B)}(D, C) =  \frac{H_b(q_X)(q_X - D)}{q_X} 
$ if $D < \frac{q_X(C - H_b(q_{S_1})}{H_b(m) - H_b(q_{S_1})}$.

\textbf{Case 2.} Constraint~(\ref{eq:class-reduced}) is active and constraint~(\ref{eq:dist-reduced-qx}) is inactive.

The constraint~(\ref{eq:class-reduced}) is active if
\begin{align*}
&a\,H_b(q_{S_1}) + (1-a)H_b(m) = C, \\
& a = \frac{H_b(m) - C}{H_b(m) - H_b(q_{S_1})}.
\end{align*}

As a consequence, we have
\[
R_{ }^{(B)}(D, C) = \frac{H_b(q_X)[H_b(m) - C]}{H_b(m) - H_b(q_{S_1})}.
\]

And,
\begin{align*}
p_U^\star(1) &= \frac{H_b(m) - C}{H_b(m) - H_b(q_{S_1})},\\
p_U^\star(2) &= 0,\\
p_U^\star(3) &= \frac{C - H_b(q_{S_1})}{H_b(m) - H_b(q_{S_1})},\\
p_U^\star(4) &= 0.
\end{align*}

The constraint~(\ref{eq:dist-reduced-qx}) is inactive if
\begin{align*}
&q_X(1-a) < D,\\
&D > \frac{q_X [C - H_b(q_{S_1})]}{H_b(m) - H_b(q_{S_1})}.
\end{align*}

Therefore, $
R_{ }^{(B)}(D, C) = \frac{H_b(q_X)[H_b(m) - C]}{H_b(m) - H_b(q_{S_1})} 
$ if $D > \frac{q_X [C - H_b(q_{S_1})]}{H_b(m) - H_b(q_{S_1})}$.

\textbf{Case 3:} Both the rate constraint~(\ref{eq:dist-reduced-qx}) and the classification constraint~\eqref{eq:class-reduced} are active.

From case 2, we know that the classification constraint~\eqref{eq:class-reduced} is active if 
\begin{align*}
a = \frac{H_b(m) - C}{H_b(m) - H_b(q_{S_1})}.
\end{align*}
And
\[
R_{ }^{(B)}(D, C) = \frac{H_b(q_X)[H_b(m) - C]}{H_b(m) - H_b(q_{S_1})}.
\]

The rate constraint~(\ref{eq:dist-reduced-qx}) is active if 
\begin{align*}
   D = \frac{q_X [C - H_b(q_{S_1})]}{H_b(m) - H_b(q_{S_1})}.
\end{align*}

Therefore, $R_{ }^{(B)}(D, C) = \frac{H_b(q_X)[H_b(m) - C]}{H_b(m) - H_b(q_{S_1})}$ if $\frac{q_X [C - H_b(q_{S_1})]}{H_b(m) - H_b(q_{S_1})}$.

\textbf{Case 4.}  Both constraint \eqref{eq:dist-reduced-qx} and constraint \eqref{eq:class-reduced} are inactive.

We observe that the rate achieves its theoretical minimum when $R_{ }^{(B)}(D, C) = H_b(q_X)\left(p_U(1) + p_U(2)\right) = 0$, which implies \( p_U(1) = p_U(2) = 0 \). Hence, we obtain $p_U(3) = 1, \quad p_U(4) = 0$.

Then, using constraints (\ref{eq:dist-reduced-qx}) and (\ref{eq:class-reduced}), we find that the feasibility of this configuration requires $D \geq q_X$ and $C \geq H_b\left(m\right)$. Therefore, the minimum achievable rate is zero, i.e., $R_{ }^{(B)}(D, C) = 0$, if and only if the distortion and classification loss exceed the respective thresholds: $D \geq q_X$ and $ C \geq H_b\left(m\right)$.

Collecting the expressions obtained across the four cases yields the stated closed-form characterization of $R_{ }^{(B)}(D, C)$, which completes the proof.

\subsection{Proof of Theorem \ref{Oneshot_Bernoulli_radom_DRC}}\label{app:proof_Oneshot_Bernoulli_radom_DRC}

We proceed in the same idea as the proof of Theorem~\ref{Oneshot_Bernoulli_radom_RDC}. In particular, the Bernoulli DRC problem under Hamming distortion admits the finite-dimensional reduction
\begin{subequations}
\begin{align}
D_{ }^{(B)}(R, C) 
&= \!\!\!\!\!\!\!\!\!\!\!\!\!\! \min_{p_U(1), p_U(2), p_U(3), p_U(4)} \!\!\!\!\!\!\!\!\!\!\!\!\!\! p_U(2) + q_X p_U(3) + (1 - q_X)p_U(4) 
\label{RDC_Bernoulli_random}\\
\textnormal{s.t.} \quad 
& H_b(q_X)(p_U(1) + p_U(2))  \leq R, \label{rate_common_constraint_1_RDC} \\
&(p_U(1) + p_U(2)) H_b(q_{S_1}) \nonumber \\
&+ (p_U(3) + p_U(4)) H_b(m) \leq C, \label{rate_common_constraint_3_RDC}\\
&p_U(1) + p_U(2) + p_U(3) + p_U(4) = 1, \label{rate_common_constraint_4_RDC}\\
&p_U(1),\, p_U(2),\, p_U(3),\, p_U(4) \geq 0 \label{rate_common_constraint_5_RDC}.
\end{align}   
\end{subequations}

\textbf{Structure and activity of the nonnegativity constraints.}
Let $a\triangleq p_U(1)+ p_U(2)$ so that $ p_U(3)+ p_U(4)=1-a$. Write $ p_U(3)=t\in[0,1-a]$, $ p_U(4)=(1-a)-t$. Then
\begin{align*}
\mathbb{E}[d_H(X,\hat{X})] &= p_U(2)+q_X  p_U(3)+(1-q_X)  p_U(4) \\
&=  p_U(2) + (2q_X-1)t + (1-q_X)(1-a).   
\end{align*}

Since $0\le q_X\le \tfrac12$, we have $2q_X-1\le 0$; hence the distortion $\mathbb{E}[d_H(X,\hat{X})]$ is minimized by taking $t=1-a$ (i.e., $p_3=1-a$, $p_4=0$). Thus, for fixed $a$,
\[
\min_{ p_U(3)+  p_U(4)=1-a}\;\bigl[ p_U(2) +q_X  p_U(3)+(1-q_X)  p_U(4)\bigr] =  p_U(2) + q_X(1-a).
\]
Because the objective increases with $ p_U(2)$ and $ p_U(2)$ only worsens the objective function, we must set $ p_U(2)^\star=0$.

Therefore, the optimization problem collapses to a single scalar $a=p_U(1)\in[0,1)$:
\begin{align}
\label{prob:RDC_Bernoulli_random_one_var}
D_{ }^{(B)}(R, C) = \min_{a}\;&  q_X(1-a)\\
\textnormal{s.t.}\;& H_b(q_X)\,a \le R, \label{dist-reduced-qx}\\
& a\,H_b(q_{S_1}) + (1-a)H_b(m) \le C, \label{class-reduced}\\
& a \geq 0.
\end{align}
and the corresponding optimizer in the original variables is $(p_U^\star(1),p_U^\star(2),p_U^\star(3),p_U^\star(4))=(a^\star,\,0,\,1-a^\star,\,0)$.

Problem~(\ref{prob:RDC_Bernoulli_random_one_var}) is a linear program in the scalar variable $a$, so the optimum is attained at a boundary point determined by which constraints are active. We next enumerate all possible activity patterns of the rate constraint~(\ref{dist-reduced-qx}) and the classification constraint~(\ref{class-reduced}) and derive the corresponding closed forms.

\textbf{Case 1.} Constraint~(\ref{dist-reduced-qx}) is active and constraint~(\ref{class-reduced}) is inactive.

The constraint~(\ref{dist-reduced-qx}) is active, we get
\begin{align*}
H_b(q_X)\,a = R \Rightarrow a = \frac{R}{H_b(q_X)}.
\end{align*}

As a consequence, we have
\[
D_{ }^{(B)}(R, C) =  \frac{q_X(H_b(q_X) - R)}{H_b(q_X)}.
\]

And,
\begin{align*}
p_U^\star(1) &= \frac{R}{H_b(q_X)},\\
p_U^\star(2) &= 0,\\
p_U^\star(3) &= \frac{H_b(q_X) - R}{H_b(q_X)},\\
p_U^\star(4) &= 0.
\end{align*}

The constraint~(\ref{class-reduced}) is inactive if
\begin{align*}
&a\,H_b(q_{S_1}) + (1-a)H_b(m) < C, \\
&\frac{R}{H_b(q_X)} H_b(q_{S_1}) + \left(1 - \frac{R}{H_b(q_X)} \right)H_b(m) < C, \\
& C > \frac{R (H_b(q_{S_1}) - H_b(m))}{H_b(q_X)}  + H_b(m).
\end{align*}

Therefore, $
D_{ }^{(B)}(R, C) =  \frac{q_X(H_b(q_X) - R)}{H_b(q_X)} 
$ if $C > \frac{R (H_b(q_{S_1}) - H_b(m))}{H_b(q_X)}  + H_b(m)$.

\textbf{Case 2.} Constraint~(\ref{class-reduced}) is active and constraint~(\ref{dist-reduced-qx}) is inactive.

The constraint~(\ref{class-reduced}) is active if
\begin{align*}
&a\,H_b(q_{S_1}) + (1-a)H_b(m) = C, \\
& a = \frac{C - H_b(m)}{H_b(q_{S_1}) - H_b(m)}.
\end{align*}

As a consequence, we have
\[
D_{ }^{(B)}(R, C) = \frac{q_X(H_b(q_{S_1}) - C)}{H_b(q_{S_1}) - H_b(m)}.
\]

And,
\begin{align*}
p_U^\star(1) &= \frac{C - H_b(m)}{H_b(q_{S_1}) - H_b(m)},\\
p_U^\star(2) &= 0,\\
p_U^\star(3) &= \frac{H_b(q_{S_1}) - C}{H_b(q_{S_1}) - H_b(m)},\\
p_U^\star(4) &= 0.
\end{align*}

The constraint~(\ref{dist-reduced-qx}) is inactive if
\begin{align*}
&H_b(q_X)a < R,\\
&H_b(q_X)\frac{C - H_b(m)}{H_b(q_{S_1}) - H_b(m)} < R, \\
&C < \frac{R (H_b(q_{S_1}) - H_b(m))}{H_b(q_X)} + H_b(m).
\end{align*}

Therefore, $
D_{ }^{(B)}(R, C) = \frac{q_X(H_b(q_{S_1}) - C)}{H_b(q_{S_1}) - H_b(m)} 
$ if $H_b(q_{S_1}) \leq C < \frac{R (H_b(q_{S_1}) - H_b(m))}{H_b(q_X)}  + H_b(m)$.

\textbf{Case 3:} Both the rate constraint~(\ref{dist-reduced-qx}) and the classification constraint~\eqref{class-reduced} are active.

From case 2, we know that the classification constraint~\eqref{class-reduced} is active if 
\begin{align*}
a = \frac{C - H_b(m)}{H_b(q_{S_1}) - H_b(m)}.
\end{align*}
And
\[
D_{ }^{(B)}(R, C) = \frac{q_X(H_b(q_{S_1}) - C)}{H_b(q_{S_1}) - H_b(m)}.
\]

The rate constraint~(\ref{dist-reduced-qx}) is active if 
\begin{align*}
   C = \frac{R (H_b(q_{S_1}) - H_b(m))}{H_b(q_X)}  + H_b(m).
\end{align*}

Therefore, $D_{ }^{(B)}(R, C) = \frac{q_X(C - H_b(q_{S_1}))}{H_b(m) - H_b(q_{S_1})}$ if $C = \frac{R (H_b(q_{S_1}) - H_b(m))}{H_b(q_X)}  + H_b(m)$.

\textbf{Case 4.}  Both constraint \eqref{dist-reduced-qx} and constraint \eqref{class-reduced} are inactive.

When \( C > H_b(q_S)  \), implying that the classification constraint~\eqref{class-reduced} is inactive, and the rate \( R \) is sufficiently large such that \( R > H_b(q_X) \), meaning the rate constraint~\eqref{dist-reduced-qx} is also inactive, the minimum achievable distortion \( D^{(B)}_{ }(R, C) \) reaches its theoretical lower bound, i.e., \( D^{(B)}_{ }(R, C) = 0 \). Therefore, $D^{(B)}_{ }(R, C) = 0$ if $ C > H_b(q_S) \text{ and } R > H_b(q_X)$.

Collecting the expressions obtained in Cases~1-4 yields the stated closed-form characterization of $D_{ }^{(B)}(R, C)$, completing the proof.

\subsection{Proof of Theorem \ref{theo:lowerBoundary}}\label{app:proof_lowerBoundary}

The lower boundary of DC achievable
region \( \Pi(p_{Z|X}) \), under the Hamming distortion, is defined as 
\begin{mini!}|s|[2] 
{p_{\hat{X}|Z}} 
{ P(X \neq \hat{X})} 
{\label{RDPC}} 
{D(C) = } 
\addConstraint{H(S | \hat{X})}{\leq C.} 
\end{mini!}

We next express the objective in terms of the parameters \(\{q_i,\epsilon_i,p_i\}\).
The Hamming distance term can be written as
\begin{align*}
\begin{split}
P(X \neq \hat{X}) &= \sum_{i=1}^n \Big[ p_{X|Z}(0|i) p_Z(i) p_{\hat{X}|Z}(1|i) \\
&+ p_{X|Z}(1|i) p_Z(i) p_{\hat{X}|Z}(0|i) \Big],\\
&= \sum_{i=1}^n \Big[ q_i (1 - \epsilon_i) + q_i (2 \epsilon_i - 1) p_i \Big].
    \end{split}    
\end{align*}

Recall that the channel between \(X\) and \(\hat{X}\) is a binary symmetric channel, so that $\hat{X} = X \oplus S_2$, where \(S_2 \sim \mathrm{Bern}(P_{\hat{X}|X}(1|0))\).
Moreover, the task variable satisfies \(S = X \oplus S_1\), with \(S_1 \sim \mathrm{Bern}(q_{S_1})\).
Combining these relations, we obtain
\[
S = X \oplus S_1 = (\hat{X} \oplus S_2) \oplus S_1 = \hat{X} \oplus (S_1 \oplus S_2),
\]
which implies that $H(S | \hat{X}) = H(S_1 \oplus S_2)$.
 
Since the modulo-\(2\) sum of two independent Bernoulli random variables is again Bernoulli, its parameter is
\begin{align*}
P(S_1 \oplus S_2 = 1) &= q_{S_1} * P_{\hat{X}|X}(1|0) \\
&= (1 - q_{S_1})P_{\hat{X}|X}(1|0) + q_{S_1}(1 - P_{\hat{X}|X}(1|0)),    
\end{align*}
    
Therefore, \( S_1 \oplus S_2 \sim \mathrm{Bern}(q_{S_1} * P_{\hat{X}|X}(1|0)) \) and $H(S | \hat{X}) = H(S_1 \oplus S_2) = H(q_{S_1} * P_{\hat{X}|X}(1|0))$.

By Mrs. Gerber's Lemma \cite{Wyner1973, Wang2024}, we have: $H(S | \hat{X}) \leq C \Rightarrow H(X | \hat{X}) \leq H(C_0)$, where \(C_0 \triangleq \frac{H^{-1}(C) - q_{S_1}}{1 - 2q_{S_1}},\) and \( H^{-1} : [0, 1] \rightarrow \left[ 0, \frac{1}{2} \right] \),
\(
C_0 = \frac{H^{-1}(C) - q_{S_1}}{1 - 2q_{S_1}} \leq \frac{\frac{1}{2} - q_{S_1}}{1 - 2q_{S_1}} = \frac{1}{2}.
\) 

If $P_{\hat{X}|X}(1|0) \leq C_0 \leq \frac{1}{2}$, then
\begin{align*}
P_{\hat{X}|X}(1|0) \leq C_0 = \frac{H^{-1}(C) - q_{S_1}}{1 - 2q_{S_1}}, \\
\Rightarrow q_{S_1} * P_{\hat{X}|X}(1|0) \leq H^{-1}(C),  \\
H(S|\hat{X}) = H(q_{S_1} * P_{\hat{X}|X}(1|0)) \leq C
\end{align*}

Hence, from $H(S | \hat{X}) = H(S_1 \oplus S_2) = H(q_{S_1} * P_{\hat{X}|X}(1|0)) \leq C$, we have: $q_{S_1} * P_{\hat{X}|X}(1|0) \leq H^{-1}(C),$
\[
(1 - 2q_{S_1}) P_{\hat{X}|X}(1|0) + q_{S_1} \leq  H^{-1}(C).
\]
We now express the crossover probability \(P_{\hat{X}|X}(1|0)\) in terms of \(\{q_i,\epsilon_i,p_i\}\):
\begin{align*}
P_{\hat{X}|X}(1|0) &= P_{\hat{X}|X}(0|1) = \frac{p_{\hat{X},X}(0,1)}{p_X(1)} \\
&= \frac{\sum_{i = 1}^{n}p_{X|Z}(1|i)p_Z(i)p_{\hat{X}|Z}(0|i)}{1 - \sum_{i=1}^{n}q_i(1 - \epsilon_i)} \\
&= \frac{\sum_{i=1}^{n} q_i \epsilon_i p_i}{1 - \sum_{i=1}^{n}q_i(1 - \epsilon_i)},    
\end{align*}

Hence, the classification constraint is as follows
\[
(1 - 2q_{S_1}) \Big(\sum_{i=1}^{n} q_i \epsilon_i p_i \Big) - (H^{-1}(C) - q_{S_1}) \Big(1 - \sum_{i=1}^{n} q_i (1 - \epsilon_i \Big) \!\! \leq 0.
\]

Putting the above together yields a linear program in \(\{p_i\}_{i=1}^n\) as represented by the optimization problem (\ref{prob:D_C_lower_boundary}).

\subsection{Proof of Theorem \ref{theo:ratePenaltyBound}}\label{app:proof_ratePenaltyBound}

Recall that the minimum rate \(R^{(\infty)}(\Theta(R))\) is given by the solution to the following optimization problem:
\begin{mini!}|s|[2] 
{p_{\hat{X}_1, \hat{X}_2 | X}} 
{I(X; \hat{X}_1, \hat{X}_2)} 
{\label{RDPC}} 
{R^{(\infty)}(\Theta(R)) =} 
\addConstraint{\!\!\!\!\!\!\!\!\!\!\! P(X \neq \hat{X}_1)}{\leq C_0} 
\addConstraint{\!\!\!\!\!\!\!\!\!\!\! P(X \neq \hat{X}_2)}{\leq b} 
\addConstraint{\!\!\!\!\!\!\!\!\!\!\! H(S | \hat{X}_1)}{\leq H(C_0 (1 - 2 q_{s_1})) + q_{s_1})} 
\addConstraint{\!\!\!\!\!\!\!\!\!\!\! H(S | \hat{X}_2)}{\leq C_{\min}.} 
\end{mini!}

To make the problem explicit, we expand the objective and constraints in terms of the conditional pmf $p_{\hat{X}_1,\hat{X}_2|X}$. The mutual information term is
\[
\begin{split}
&I(X; \hat{X}_1, \hat{X}_2) = \sum_{x, \hat{x}_1, \hat{x}_2 \in \{0, 1\}} p(x, \hat{x}_1, \hat{x}_2) \log \frac{p(x, \hat{x}_1, \hat{x}_2)}{p(x) p(\hat{x}_1, \hat{x}_2)},\\
&= (1 - q) \sum_{i, j \in \{0, 1\}} p_{ij|0} \log \frac{p_{ij|0}}{(1 - q) p_{ij|0} + q p_{ij|1}}\\
&+ q \sum_{i, j \in \{0, 1\}} p_{ij|1} \log \frac{p_{ij|1}}{(1 - q) p_{ij|0} + q p_{ij|1}}.
\end{split}
\]

\begin{figure*}
\begin{align}
\label{lowerBoundObjective}
&I(X; \hat{X}_1, \hat{X}_2) \nonumber \\
&\geq I_{LB}(X; \hat{X}_1, \hat{X}_2) = (1 - q) \nonumber \\ 
&\Bigg[ p_{00|0} \log \frac{p_{00|0}}{(1 - q) p_{00|0} + q p_{00|1}} + p_{01|0} \log \frac{p_{01|0}}{(1 - q) p_{01|0} + q p_{01|1}} + p_{11|0} \log \frac{p_{11|0}}{(1 - q) p_{11|0} + q p_{11|1}} \Bigg]  \\
&+ q \Bigg[ p_{00|1} \log \frac{p_{00|1}}{(1 - q) p_{00|0} + q p_{00|1}} + p_{01|1} \log \frac{p_{01|1}}{(1 - q) p_{01|0} + q p_{01|1}} + p_{11|1} \log \frac{p_{11|1}}{(1 - q) p_{11|0} + q p_{11|1}} \Bigg]. \nonumber
\end{align}
\end{figure*}

The Hamming distortion constraints can be written as
\[
\begin{split}
P(X \neq \hat{X}_1) &= p_{\hat{X}_1 | X}(1 | 0) p_X(0) + p_{\hat{X}_1 | X}(0 | 1) p_X(1),\\
&= (1 - q)(p_{10|0} + p_{11|0}) + q(p_{00|1} + p_{01|1}).
\end{split}
\]
\[
\begin{split}
P(X \neq \hat{X}_2) &= p_{\hat{X}_2 | X}(1 | 0) p_X(0) + p_{\hat{X}_2 | X}(0 | 1) p_X(1),\\
&= (1 - q)(p_{01|0} + p_{11|0}) + q(p_{00|1} + p_{10|1}).
\end{split}
\]

We next translate the two classification constraints into algebraic inequalities. Similarly, assume the channel of $X$ and $\hat{X}_1$ is binary symmetric channel: $\hat{X}_1 = X \oplus S_2$, $S_2 \sim \text{Bern}(P_{\hat{X}|X}(1|0))$. We have, \( S = X \oplus S_1 \) with \( S_1 \sim \mathrm{Bern}(q_{s_1}) \). Therefore, 
\begin{align*}
H(S | \hat{X}_1) &= H(S_1 \oplus S_2) = H(q_{s_1} * P_{\hat{X}_1|X}(1|0)) \\
&\leq H(C_0 (1 - 2 q_{s_1})) + q_{s_1}),
\end{align*}
\begin{align*}
(q_{s_1} * P_{\hat{X}_1|X}(1|0) &\leq H^{-1}(H(C_0 (1 - 2 q_{s_1})) + q_{s_1})) \\
&= C_0 (1 - 2 q_{s_1})) + q_{s_1},    
\end{align*}
\begin{align*}
(1 - q_{s_1})(p_{10|0} + p_{11|0}) + q_{s_1} (1 - p_{10|0} - p_{11|0}) \\
\leq C_0 (1 - 2 q_{s_1})) + q_{s_1}.   
\end{align*}

The constraint $H(S | \hat{X}_2) \leq C_{\min}$ implies that
\begin{align*}
(1 - q_{s_1})(p_{01|0} + p_{11|0}) + q_{s_1} (1 - p_{01|0} - p_{11|0}) \\
\leq H^{-1}(C_{\min}).   
\end{align*}

Building on the framework developed in~\cite{Qian2023RateDistortionPerception}, we derive lower and upper bounds for the rate penalty $R^{(\infty)}(\Theta(R))$.
 
\textbf{The Lower Bound of Rate.} The objective function is not differentiable at the boundary points, i.e., at points where some \( p_{ij|k} = 0 \). According to the log-sum inequality \cite{Cover2006}, the tight lower bound of the objective function is derived as (\ref{lowerBoundObjective}). Then, the resulting relaxation leads to the linear program (\ref{R_LB}). 

\textbf{The Upper Bound of Rate.} We now obtain a complementary upper bound using the geometry of the level sets of $R^{(\infty)}(D,C)$. 
Observe that the RDC function $R^{(\infty)}(D,C)$ is convex~\cite{Wang2024}, and therefore each level curve $C(D)$ defined by $R^{(\infty)}(D,C(D))=\hat{R}$ is convex as well. 

By the implicit function theorem, the slope along a level curve satisfies $C'(D) = - \frac{\partial R^{(\infty)}(D,C(D))/\partial D}{\partial R^{(\infty)}(D,C(D))/\partial C}$. By convexity, $C(D)$ lies above its supporting line at $D=b$, namely the affine map
$D \mapsto C'(b)(D-b)$. Consequently, to guarantee universality over $\Theta(R)$ it is sufficient that the entire segment between $(C_0,\, C'(b)(C_0-b))$ and $(b,\,0)$ is achievable. 
Moreover, since the achievable set is convex, it suffices to verify achievability of the two endpoints $(C_0,\, C'(b)(C_0-b))$ and $(b,\,0)$, which then implies achievability of the full segment. Substituting these two endpoints into the RDC function yields an explicit upper bound, denoted by $R^{(\infty)}_{UB}(\Theta(R))$.

Recall from Theorem~\ref{RDC_Wang2024} that
\begin{align*}
R^{(\infty)}(D,C)= \begin{cases}
H(b) - H(D), \hfill 0 \leq D \leq C_0 \text{ and } D\leq b,\\
H(b) - H(C_0), \hfill C_0 < D \leq b  \text{ and } C_0\leq b,\\
0, \hfill \min\{D,C_0\}> b.
\end{cases}
\end{align*}

In particular, when $D \in [C_0,b]$, the slope of the tangent for the upper bound is
\begin{align*}
\kappa_{UB} = C'(b)  
= \left. \frac{\partial C}{\partial D} \right|_{D=C_{0},\, C = C_{\min}} 
= - \left. \frac{\tfrac{\partial R}{\partial D}}{\tfrac{\partial R}{\partial C}} \right|_{D=C_{0},\, C \!\! = C_{\min}} \!\!\!\!\!\! = 0.    
\end{align*}
Hence, $H^{-1}(C^{'}(b)(C_0 - b)) = 0$. Finally, the upper bound $R^{(\infty)}_{UB}(\Theta(R))$ can be obtained by solving the linear programming problem (\ref{R_UB}).

\end{document}